\documentclass[letterpaper]{article}
\usepackage[draft]{aaai2027}      
\usepackage{times}
\usepackage{helvet}
\usepackage{courier}
\usepackage[hyphens]{url}
\usepackage{natbib}
\usepackage{amsmath,amssymb,amsthm}
\usepackage{booktabs}
\usepackage{enumitem}
\usepackage{xcolor}

\usepackage{bm}
\usepackage[f]{esvect}
\usepackage{tikz}
\usepackage{array}
\usepackage{xparse}
\usepackage{xspace}
\usepackage{mathrsfs}
\usepackage{mathtools}
\usepackage{booktabs}
\usepackage{algorithm}
\usepackage[noend]{algpseudocode}
\usepackage[capitalise]{cleveref}

\theoremstyle{plain}
\newtheorem{theorem}{Theorem}
\newtheorem{lemma}[theorem]{Lemma}
\newtheorem{corollary}[theorem]{Corollary}
\newtheorem{proposition}[theorem]{Proposition}

\theoremstyle{definition}
\newtheorem{definition}[theorem]{Definition}
\theoremstyle{remark}

\DeclareMathOperator*{\argmin}{arg\,min}

\providecommand{\set}[1]{\ensuremath{\{ #1 \}}\xspace}

\providecommand{\setbuild}[2]{\ensuremath{\set{#1 \mid #2}}\xspace}
\providecommand{\tuple}[1]{\ensuremath{( #1 )}\xspace}

\DeclareMathSymbol{\R}{\mathord}{AMSb}{"52}

\newcommand{\transpose}{^\mathsf{T}}

\providecommand{\Expect}[2][]{\ensuremath{%
\ifthenelse{\equal{#1}{}}{\mathbb{E}}{\mathbb{E}_{#1}}%
\left[#2\right]}\xspace}

\renewcommand{\th}{^\text{th}}
\newcommand{\cond}{\, | \,}

\newcommand{\xvec}{\ensuremath{\mathbf{x}}\xspace}

\NewDocumentCommand{\xcomp}{m g}{%
  x_{#1\IfNoValueF{#2}{#2}}%
}

\newcommand{\Mmat}{\ensuremath{\mathbf{M}}\xspace}
\newcommand{\Mrow}[1]{\Mmat_{#1}}
\newcommand{\Mcol}[1]{\Mmat\transpose_{#1}}
\newcommand{\Mcomp}[2]{M_{#1#2}}

\newcommand{\wmat}{\ensuremath{\mathbf{w}}\xspace}
\newcommand{\wvec}{\ensuremath{\mathbf{w}}\xspace}
\newcommand{\wrow}[1]{\wvec_{#1}}
\newcommand{\wcol}[1]{\wvec\transpose_{#1}}

\NewDocumentCommand{\wcomp}{m g}{%
  w_{#1\IfNoValueF{#2}{#2}}%
}

\newcommand{\pvec}{\ensuremath{\mathbf{p}}\xspace}
\newcommand{\pcomp}[1]{p_{#1}}
\newcommand{\qvec}{\ensuremath{\mathbf{q}}\xspace}
\newcommand{\qcomp}[1]{q_{#1}}
\newcommand{\pq}{$(\pvec, \qvec)$}

\newcommand{\elec}{\mathcal{E}}
\newcommand{\V}{\ensuremath{V}\xspace}
\newcommand{\C}{\ensuremath{C}\xspace}
\newcommand{\prof}{{\vv{\succ}}}
\renewcommand{\succeq}{\succcurlyeq}
\newcommand{\election}{$\elec = \tuple{\V, \C, \prof}$\xspace}

\newcommand{\cons}{\sim} 
\newcommand{\cost}{\operatorname{SC}}
\newcommand{\barcost}{\overline{\operatorname{SC}}}
\newcommand{\poscost}[2]{\cost_{#1}(#2)}

\renewcommand{\top}{\operatorname{top}} 
\newcommand{\plu}{\operatorname{plu}} 
\newcommand{\uni}{\operatorname{uni}} 

\newcommand{\sdom}[1]{\mathcal{S}_{#1}} 
\newcommand{\wdom}{\mathcal{W}^{n \times m}} 
\newcommand{\widom}{\mathcal{W}^{n \times m}_{\operatorname{id}}} 
\newcommand{\wsumdom}{\mathcal{W}^{n \times m}_{\operatorname{sum}}} 
\newcommand{\wtopdom}{\mathcal{W}^{n \times m}_{\operatorname{top}}} 

\newcommand{\dist}{\operatorname{Dist}}
\newcommand{\distid}{\dist_{\operatorname{id}}}
\newcommand{\distsum}{\dist_{\operatorname{sum}}}
\newcommand{\disttop}{\dist_{\operatorname{top}}}

\newcommand{\fracveto}{\textsc{FractionalVeto}\xspace}

\newcommand{\pluveto}{\textsc{PluralityVeto}\xspace}
\newcommand{\copeland}{\textsc{Copeland}\xspace}

\newcommand{\drop}[2]{\delta_{#1#2}}
\newcommand{\range}[1]{\operatorname{range}(#1)} 

\title{Metric Distortion of Social Welfare Functions}

\author{
    Fatih Erdem Kizilkaya\textsuperscript{\rm 1}\thanks{Corresponding author ({fatih.erdem.kizilkaya@gmail.com}).}\thanks{This work was supported in part by the Ministry of National Education of the Republic of T\"urkiye.}, 
    Aaryaman Aggarwal\textsuperscript{\rm 2} \&
    Evi Micha\textsuperscript{\rm 1}
}
\affiliations{
    \textsuperscript{\rm 1} 
    University of Southern California \\
    \textsuperscript{\rm 2}
    Indian Institute of Technology Kanpur
}

\usepackage[suppress]{color-edits}
\addauthor[Fatih]{fk}{magenta}
\addauthor[Aaryaman]{aa}{cyan}
\addauthor[Evi]{em}{orange}

\begin{document}
\maketitle

\pagestyle{plain}      
\thispagestyle{plain}  

\begin{abstract}

Metric distortion has primarily been studied for social choice functions, which select a single winner from ordinal preferences.
We extend this framework to social welfare functions,  which output  
\fkreplace{a consensus ranking}{an entire ranking over a set of $m$ candidates}. 
We~associate each voter~$v$ with
a monotone weight vector $\fkedit{\wrow{v}} =  (w_{v1},\ldots,w_{vm})$ \fkedit{that specifies the relative importance of the $i\th$ position for voter~$v$,} and define the cost of a ranking as the position-weighted sum of the~voter's distances to the ranked candidates. 
This model generalizes both single-winner voting and committee selection
while allowing different positions to carry different weights. 
We consider three information regimes \fkedit{within this model}. 

First, we study the setting in which the positional weight  \linebreak vectors are known to the rule. 
A natural approach recursively \linebreak applies a single-winner rule with distortion~$\beta$ to construct the ranking one position at a time.
We show that this yields distortion at most~$3\beta$ in general. 
This gap is not merely an artifact of the analysis: we exhibit instances where making \linebreak a locally $\beta$-approximate choice at every position drives the overall distortion up to $2\beta$, 
showing that no analysis based solely on per-round guarantees can certify a bound better \linebreak than $2\beta$.
By exploiting structural properties 
specific to \textsc{FractionalVeto} of \fkreplace{Kizilkaya and Kempe}{\citet{PluralityVeto}}, we~show 
that its~recursive extension achieves the optimal  distortion of~$3$.

Second, when all voters share the same (unknown) weight vector, recursively applying any \fkreplace{single-winner rule}{social choice function} with distortion~$\beta$ achieves distortion at most $1+(\beta-1)\range{\wvec}$,
where $\range{\wvec}=(\wcomp{1}-\wcomp{m})/\wcomp{1}$ denotes the normalized range of the common weight vector~$\wvec$. 
Thus, the guarantee interpolates between distortion~$1$ for uniform weights and~$\beta$ for single-winner weights. 
Finally, we study unknown heterogeneous weights. Without further assumptions, every rule has unbounded distortion. 
We therefore consider two natural normalizations: unit-sum, where each voter distributes one unit of value across the ranking, and unit-top, where every voter assigns unit value to the first position. 
Under both models, we~\fkreplace{establish matching asymptotic upper and lower bounds, showing}{show} that the optimal distortion is~$\Theta(m)$.

\end{abstract}

\section{Introduction}

Computational social choice~\cite{BCULP:social-choice} studies how to aggregate the preferences of multiple agents into a collective decision. 
While the classical axiomatic approach \pagebreak \linebreak  compares voting rules through the normative properties they~satisfy \cite{Arr90}, a complementary quantitative \linebreak approach assumes that \fkreplace{agents}{voters} have underlying cardinal \linebreak \fkreplace{values}{utilities} for the candidates, but that only the ordinal rankings \linebreak \fkedit{that} these \fkreplace{values}{utilities}  induce can be elicited. 
The distortion of \linebreak a voting rule measures the worst-case loss in social welfare \linebreak resulting from this~limited information~\cite{procaccia:rosenschein:distortion}.

In the most general setting, voters may have arbitrary \linebreak underlying \fkreplace{cardinal values}{utilities} consistent with their ordinal rankings. \linebreak
Without further assumptions on these \fkreplace{values}{utilities}, obtaining  \linebreak meaningful approximation guarantees is not possible, \linebreak 
motivating the study of more structured preference models. 
One~such approach is the \emph{\fkreplace{utilitarian}{normalized} distortion framework}, \linebreak
which assumes arbitrary nonnegative cardinal utilities, \linebreak
normalized so that each voter's utilities sum to one~\cite{BCHLPS:utilitarian:distortion,caragiannis:procaccia:voting,ebadian:kahng:peters:shah:optimized-distortion,couvreur:bressler}. 
A more structured model that has received significant attention in  recent years is the \emph{metric distortion framework} \cite{anshelevich:bhardwaj:elkind:postl:skowron}.
In~this~model, voters and candidates are embedded in \fkreplace{an unknown}{a latent} metric space, and each voter ranks the candidates according to their distance from them. 

This \fkreplace{framework}{notion} has led to a rich line of work on the distortion \linebreak of voting rules~\cite{anshelevich:bhardwaj:postl,anshelevich:postl:randomized,DistortionDuality,DistortionCommunication}. 
The \fkreplace{single-winner setting}{case of single-winner voting} is now well understood: deterministic \linebreak
\fkreplace{ voting rules}{social choice functions} achieving the optimal distortion of~$3$ have been developed \citep{gkatzelis:halpern:shah:resolving,PluralityVeto,VetoCore}.
More recently, this line of research has been extended to multiwinner elections, \linebreak
leading to a growing understanding of the metric distortion of committee-selection rules
under a variety of voter cost functions \citep{GoelHK18,CaragiannisSV22,KalayciKK24,burkhardt2024low}.

Many applications, however, require not a single winner \linebreak or an unordered committee, but a complete ranking of \linebreak the~candidates.
Recommendation  systems aggregate rankings \linebreak induced by different criteria, 
hiring committees rank candidates so that additional offers can be made if earlier candidates decline, 
funding agencies prioritize proposals for \linebreak sequential funding,
and universities maintain ranked waitlists \linebreak for admissions.

In these settings, different positions of the ranking may carry different importance. 
For example,  users of a product ranking typically care only about the top few positions~\cite{joachims2017unbiased}.

The metric distortion of \emph{social welfare functions}, rules that return a complete ranking rather than a single winner, has, to the best of our knowledge, not been studied before. \linebreak
Addressing
this question first requires specifying how voters evaluate an entire ranking. 
In particular, different positions of the ranking may carry different importance, \fkreplace{and this importance}{which} may vary across voters. 
For example, in a recommendation system, \linebreak some users may be willing to browse several recommendations before finding a suitable item, whereas others may rely almost exclusively on the first recommendation. \linebreak
We~capture this by associating each voter $v$ with a monotone, \linebreak non-increasing weight vector ${w_{v1}\ge \cdots \ge w_{vm}\ge 0}$,
where $w_{vi}$ specifies the importance of position~$i$ to voter~$v$. 
The~cost of a ranking is then 
the sum, over all voters~$v$ and positions~$i$, of $v$'s distance to the candidate in position~$i$, weighted by~$w_{vi}$. 
This position-weighted model was introduced by \citet{benade2019low} in the  utilitarian distortion framework and subsumes both single-winner elections, 
in which every voter values only the top-ranked candidate, 
and committee selection, 
in which every voter assigns equal value to the first $k$ positions and zero value to the remaining positions. 
In this paper, we study the metric distortion of ordinal social welfare functions under this general model.

\paragraph{Contributions.}
Without any assumptions on the weights, or without revealing them to the social welfare function, bounded metric distortion is impossible. 
This motivates the~study of three natural information regimes, distinguished by what is known or assumed about the positional weights. Our results are summarized in Table~\ref{tab:results}.

\emph{1) Known Weights.}
First, we consider the setting in which the positional weight vectors are given as input to the social \linebreak welfare function. 
In this case, a natural approach is to \linebreak construct  the ranking recursively, 
filling the positions one by one by repeatedly applying a single-winner rule to the \linebreak remaining candidates. 
However, the distortion guarantee of the underlying single-winner rule does not automatically \linebreak carry over to its recursive extension.
In particular, the~best generic guarantee that we obtain for recursively applying \linebreak an  arbitrary \mbox{distortion-$\beta$} rule is a distortion of $3 \beta$,
while we~exhibit instances in which a local distortion of~$\beta$ at every recursive~step results in an~overall distortion of~$2\beta$.

Surprisingly, our main result shows that this apparent loss is not inherent.
We prove that recursively applying \fracveto, a generalization of the optimal deterministic single-winner rule \pluveto \cite{VetoCore}, yields a social welfare function with the optimal  distortion of~$3$ for every monotone weight profile.
Thus, constructing an entire ranking incurs no additional distortion over the classical single-winner setting.
Our proof exploits structural properties specific to \fracveto, \fkreplace{rather than merely its single-winner distortion guarantee}
{beyond its~single-winner distortion guarantee}. 
We also study \fkedit{Weighted} Recursive \copeland.
Although \copeland achieves distortion~$5$ in the single-winner setting,
here we~obtain \fkreplace{a distortion bound of~$7$}{a bound of~$7$}, \linebreak leaving open whether the~optimal bound of~$5$ can 
\fkreplace{still be achieved}{be achieved}.

\emph{2) Identical Weights.}
Second, we study the \emph{identical weights} setting, in which the positional weights are unknown but all voters share the same weight vector. 
We show that recursively applying any single-winner rule with distortion $\beta$ achieves distortion at most $1+(\beta-1)\range{\wvec}$, 
where $\range{\wvec}=(\wcomp{1}-\wcomp{m})/\wcomp{1}$ denotes the normalized~range of the common weight vector~$\wvec$. 
Thus, the distortion \linebreak interpolates continuously between~$1$, when all positions are valued equally, and~$\beta$, when only the first position matters. 
This~strictly generalizes previous recursive guarantees for committee selection and is nearly matched by our \fkreplace{lower bounds}{lower-bound construction}.

\emph{3) Normalized Weights.}
Finally, we consider the \emph{normalized weights} setting, in which voters may have different weight vectors but these vectors satisfy a natural normalization.
We study two such normalizations: unit-sum, where each voter distributes one unit of value across the ranking, and unit-top, where every voter assigns unit value to the~first position. 
We show that, under both models, the optimal \linebreak distortion is $\Theta(m)$: Recursive \pluveto\ achieves a~linear upper bound, and we establish matching asymptotic lower bounds.

\fkreplace{\begin{table}[t]
	\centering
	\small
	\begin{tabular}{@{}ll@{}}
		\toprule
		Setting & Distortion \\ \midrule
		Known weights  (Weighted Recursion) &
		$3\beta$ (optimal) \\
		
		Known weights  (\fracveto) &
		$3$ (optimal) \\
		
		Known weights (\copeland) &
		$5\le \mathrm{Dist}\le7$ (conj.\ $5$) \\
		
		\fkedit{Unknown} identical weights &
		$\le 1+(\beta-1)\range{\wvec}$ \\
		
		Unknown normalized weights &
		$\Theta(m)$:
		$\left[\frac{m}{2},\,8m+3\right]$ \\
		\bottomrule
	\end{tabular}
	\caption{Summary of our results. Here, $\beta$ denotes the distortion of the underlying single-winner rule, and $\range{\wvec}=(\wcomp{1}-\wcomp{m})/\wcomp{1}$ is the normalized range of the common weight vector.}
	\label{tab:results}
\end{table}}
{\begin{table}[t]
	\centering
	\small
	\resizebox{\columnwidth}{!}{%
		\begin{tabular}{@{}lll@{}}
			\toprule
			Weight Regime & Recursive rule & Distortion \\ \midrule
			Known     & Weighted $f$         & $\leq 3\beta$ \\
            Known       & \textsc{FractionalVeto}            & $3$ (optimal) \\
			    	Known        & Weighted \copeland   & $5\leq\mathrm{Dist.}\leq7$ \\
                        \addlinespace[1pt]
			Identical & $f$                & $\leq1+(\beta-1)\range{\wvec}$ \\
            \addlinespace[1pt]
			Normalized (unit-sum) & \pluveto            & $\Theta(m)$: $\left[\frac{m}{2},\,8m+3\right]$ \\
			Normalized (unit-top) & \pluveto            & $\Theta(m)$: $\left[\frac{m}{4},\,4m+1\right]$ \\

			\bottomrule
		\end{tabular}%
	}
	\caption{The summary of our results. 
	Here, $\beta$ denotes the~distortion of the underlying social choice function $f$, and
	$\range{\wvec}=(\wcomp{1}-\wcomp{m})/\wcomp{1}$ is the normalized range of the common weight vector $\wvec$.}
	\label{tab:results}
\end{table}}

\paragraph{Related Work.}

Distortion has been studied extensively under both utility-based and metric preference models, as well as under  richer~\cite{ebadian2025every, mandal2019efficient, DistortionCommunication} or more limited preference information~\cite{borodin2022distortion, ebadian2024metric, anagnostides2026sampling}, and in other settings such as matching~\cite{anari2023distortion, amanatidis2022few}.

A work particularly relevant to ours is that of \citet{PluralityVeto}, who introduced \textsc{FractionalVeto}, a fractional generalization of \pluveto. Since this rule plays a central role in our algorithms and analysis, we review it in detail in the Preliminaries. 
Our work is also closely related to the metric committee-selection setting of~\citet{GoelHK18}. Their objective corresponds to the special case of our identical-weights model in which every voter assigns equal value to the first $k$ positions and zero value thereafter. They show that recursively applying any single-winner rule with distortion~$\beta$ preserves distortion~$\beta$ in this setting. Our result strictly generalizes theirs to arbitrary monotone shared weight vectors.
\pagebreak

Another conceptually relevant work  is that of~\citet{benade2019low}, who introduced the position-weighted objective for extending distortion from social choice to social welfare functions in the utilitarian setting. They show that randomized social welfare functions can asymptotically match the distortion of the corresponding social choice problem, even when the positional weights are unknown. We study the same objective in the metric distortion framework, obtaining a fundamentally different picture: knowledge of the positional weights becomes crucial for achieving bounded distortion.


\section{Preliminaries} \label{sec:prelims}

We write vectors and matrices with bold letters.
We denote the $i\th$ entry of a vector $\xvec$ by $\xcomp i$,
and we extend this notation to sets via addition, i.e., $\xcomp T = \sum_{i \in T} \xcomp{i}{}$.
Given a matrix $\Mmat$, we denote the $i\th$ row vector by $\Mrow{i}$, 
the $j\th$ column vector by $\Mcol j$, 
and the entry in the $i\th$ row and $j\th$ column by~$\Mcomp{i}{j}$.
Given a set $X$, let $\Delta_X$ denote the set of non-negative weight vectors over $X$ that add to~1.
A \emph{ranking} $\sigma = (\sigma_1, \ldots, \sigma_k)$ 
of a set $X$ is a permutation of its elements where $\sigma_i$ occupies \emph{position}~$i$, and $\sdom{X}$ denotes the set of all such~rankings.

\paragraph{Elections.}
An election \election consists of 
a~set of $n$~\emph{voters} $\V$, a~set of $m$~\emph{candidates} $\C$ and
the \emph{rankings} $\prof=(\succ_v)_{v\in V}$ where $\succ_v \in \sdom{\C}$ expresses the ordinal preferences of voter $v$ over candidates.
We write $a \succ_v b$ to express that voter $v$ \emph{prefers} candidate~$a$ to candidate~$b$; 
we write  ${a \succeq_v b}$ if $a = b$ or $a \succ_v b$, and say that $v$ \emph{weakly~prefers} $a$ to $b$.
We write $\top(v)$ for the candidate ranked highest by voter~$v$.
The \emph{plurality score} of  a candidate $c$, denoted by $\plu(c)$, is the number of voters $v$ such that $\top(v)=c$.

A \emph{social choice function} $f$ maps each election $\elec$ to a
winning candidate $f(\elec)\in\C$.
A \emph{social welfare function} $F$ maps each election $\elec$ to a
consensus ranking $F(\elec)\in\sdom{\C}$.

Given a social choice function $f$, \emph{Recursive $f$} is the social welfare function that constructs a ranking $\sigma=(c_1,\ldots,c_m)$ successively, from left to right, as follows:
At position~$j$, let $A_j:=\C\setminus\set{c_1,\ldots,c_{j-1}}$ denote the remaining candidates, 
and, with a slight abuse of notation, let $\top_j(v)$ denote the candidate ranked highest in $A_j$ by voter~$v$.
The rule applies $f$ to the election restricted to $A_j$, places the resulting winner $c_j$ in position~$j$, and continues.
We use this notation for every recursive construction below.

\subsubsection{Metric Distortion.}

A \emph{metric} over a~set $X$ is a~function $d : X \times X \rightarrow \mathbb{R}_{\geq 0}$ with the following three properties for all $x, y, z \in X$: 
(1)~\emph{Identity:} $d(x,y) = 0$ if and only if $x = y$,\footnote{We allow voters and candidates to be co-located, so that their distance may be zero, i.e., we~technically consider a pseudometric.} 
(2)~\emph{Symmetry:} $d(x, y) = d(y, x)$, 
(3)~\emph{Triangle Inequality:} $d(x,y) + d(y,z) \geq d(x,z)$.

The key assumption in the notion of metric distortion is that preferences are induced by some metric $d$ over $\V \cup \C$, not available to the social choice (or welfare) function.
Given an election \election, we say that a metric $d$ is \emph{consistent} with $\prof$, and write $d \cons \prof$, if $d(v, c) \leq d(v, c')$ for all $v \in \V$ and $c, c' \in \C$ such that ${c \succeq_v c'}$.


The \emph{social cost of candidate $c \in \C$} under metric $d \cons \prof$ is $\cost(c \cond d):=\sum_{v\in V}d(v,c)$.
Let $c^*_d \in \argmin_{c\in C}\cost(c \cond d)$ denote an optimal candidate under $d$. 

\begin{definition}
	The \emph{distortion of candidate $c \in \C$} in election \election, denoted by $\dist_\elec(c)$,  
	is the largest possible ratio between the cost of $c$ and that of an optimal candidate over all metrics:
	\[
	  \dist_{\elec}(c):=
	  \sup_{d \cons \prof}
	  \frac{\cost(c \cond d)}{\cost(c_d^* \cond d)}.
	\]
	The \emph{distortion of social choice function $f$} is its worst-case distortion over all elections: $\dist(f):=\sup_{\elec}\dist_{\elec}\bigl(f(\elec)\bigr)$.
\end{definition}

\paragraph{Our Model.}
To extend metric distortion to social welfare functions, we associate each voter $v$ with a monotonically decreasing,  non-negative weight vector $\wvec_v \in \R^m$, where $\wcomp{v}{i}$ specifies the relative importance of the $i\th$ position for voter $v$
in the consensus ranking.

A \emph{weight profile} bundles the voters' weight vectors into 
an $n\times m$ matrix $\wmat=(\wvec_v)_{v\in V}$.
Let
\[
	\wdom := \setbuild{\wmat\in\mathbb R_{\geq0}^{n\times m}}
	{\wcomp{v}{1}\ge\cdots\geq \wcomp{v}{m} \text{ for all } v}
\] 
denote the domain of all such weight profiles. 

We then define the \emph{social cost of ranking~${\sigma \in \sdom{\C}}$}
under metric ${d \cons \prof}$ 
and weight profile $\wmat \in \wdom$ as
\[
  \cost\big(\sigma \cond d, \wmat\big)
  := \sum_{v\in V} \sum_{j=1}^m \wcomp{v}{j} d(v,\sigma_j).
\]

We~drop $d$~and~$\wmat$ from the notation whenever they are clear from  context, 
and write $\cost(\sigma)$ for $\cost(\sigma\cond d,\wmat)$.
With a slight abuse of notation, we write 
\(
  \poscost{j}{c}
  :=\cost(c\cond d,\wcol{j})
  =\sum_{v\in\V}\wcomp{v}{j}d(v,c)
\)
to denote the marginal social cost of assigning candidate $c$ to position~$j$.
Note that  the social~cost of a ranking $\sigma \in \sdom{C}$ decomposes as
\(
  \cost(\sigma)
  =\sum_{j=1}^m\poscost{j}{\sigma_j}.
\)
Also, let ${\sigma^*_{d, \wmat} \in \argmin_{\sigma \in \sdom{C}}\cost(\sigma \cond d, \wmat)}$ \linebreak
denote an optimal ranking under $d \cons \prof$ and $\wmat \in \wdom$.

\begin{definition}
	The \emph{distortion of ranking $\sigma \in \sdom{\C}$} in election \election, denoted by $\dist_\elec(\sigma)$,  
	is the largest possible ratio between the cost of $\sigma$ and that of an optimal ranking over all consistent metrics and weight profiles:
	\[
	  \dist_{\elec}(\sigma):=
	  \smashoperator{\sup_{\substack{
	    d\cons\prof\\
	    \quad \: \: \: \wmat\in\wdom
	  }}}\; 
	  \frac{\cost(\sigma\cond d,\wmat)}
	       {\cost(\sigma^*_{d,\wmat}\cond d,\wmat)}.
	\]
	The \emph{distortion of social welfare function $F$} is its worst-case distortion over all elections: $\dist(F):=\sup_{\elec}\dist_{\elec}\bigl(F(\elec)\bigr)$.
\end{definition}

Without further assumptions on the weight profile, 
every social welfare function $F$ that does not observe the weight profile has unbounded distortion.
Consider an election with two candidates $a, b$ and two voters $1, 2$ embedded on the real line. 
Candidate $a$ is co-located with voter~1 at position~0, 
while candidate $b$ is co-located with voter~2 at position~1. 
The~induced preferences are therefore ${a \succ_1 b}$ and ${b \succ_2 a}$. 
By symmetry, suppose that $F$ returns $\sigma=(a,b)$. 
Under the~weight profile with $\wrow{1}=(1,0)$ and $\wrow{2}=(x,0)$, 
the cost of $\sigma$ is $x$, whereas the optimal ranking $(b,a)$ has cost~1. 
Hence, the~distortion equals $x$ and is unbounded as $x\to\infty$.

\paragraph{Information Regimes.}
Obtaining bounded distortion \linebreak requires either giving the weight profile to the social welfare function or restricting the admissible weight profiles.
Thus, we consider three regimes below:

\emph{1) Known Weights.}
The weight profile $\wmat$ is~given to the social welfare function $F$, 
whose output may thus depend on both inputs, and is denoted by $F(\elec,\wmat)$.
Then, its distortion is defined accordingly as
\(
  \dist(F) := \sup_{\elec, \wmat \in \wdom}
  \dist_{\elec, \wmat}\bigl(F(\elec,\wmat) \bigr)
\)
where
\[
  \dist_{\elec, \wmat}(\sigma ) := \sup_{d\cons\prof}
  \frac{\cost(\sigma\cond d,\wmat)}
       {\cost(\sigma^*_{d,\wmat}\cond d,\wmat)}.
\]

\emph{2) Identical Weights.}
The weight profile is not given but all voters share the same weight vector.
For this regime, let 
\[
  \widom := \setbuild{\wmat\in\wdom}
  {{\wrow{u}=\wrow{v}}\text{ for all }u,v}
\]
denote the domain of admissible weight profiles,
and define distortion as
\(
  \distid(F) := \sup_{\elec, \wmat \in \widom} 
  \dist_{\elec, \wmat}\bigl(F(\elec)\bigr)
\).

\emph{3) Normalized Weights.}
The weight profile is not given but weight vectors are normalized.
We consider two types of normalizations: \emph{unit-sum} and \emph{unit-top}, 
for which the domain of admissible weight profiles and distortion, is defined as follows, respectively: 
\begin{itemize}
	\item[--] Let
	\(
		\wsumdom := \setbuild{\wmat\in\wdom}
		{\wrow{v} \in \Delta_m \text{ for all }v}
	\)
	and 
	\(
	    \distsum(F) :=
	    \sup_{\elec, \wmat\in\wsumdom}
	    \dist_{\elec,\wmat}\bigl(F(\elec)\bigr)
	\)
	\item[--] Let
	\(
		\wtopdom :=
		\setbuild{\wmat\in\wdom}
		{\wcomp{v}{1}=1\text{ for all }v}
	\)	
	and 
	\(
	    \disttop(F) :=
	    \sup_{\elec, \wmat\in\wtopdom}
	    \dist_{\elec,\wmat}\bigl(F(\elec)\bigr)
	\).
\end{itemize}

Lastly, we review concepts from the single-winner setting that will be useful also in our rank aggregation setting.

\subsubsection{Single-Winner Review.}

Given an election \election, the~\emph{domination graph} of a candidate $a$ is the bipartite graph $G_a=(\V,\C,E_a)$ in which
$(v,c)\in E_a$ if and only if~$a\succeq_v c$.
Given normalized weight vectors $\pvec\in\Delta_{\V}$ and $\qvec\in\Delta_{\C}$, 
a~non-negative matrix $\Mmat \in\R_{\geq0}^{\V\times\C}$ is a \emph{$(\pvec,\qvec)$-matching} if
\fkreplace{\[
  \sum_{c\in\C}\Mcomp{v}{c} = \pcomp{v}
  \quad\text{and}\quad
  \sum_{v\in\V}\Mcomp{v}{c} = \qcomp{c}
\]
for every $v\in\V$ and $c\in\C$.}
{
  ${\sum_{c\in\C}\Mcomp{v}{c} = \pcomp{v}}$ for all voters $v\in\V$ and
  ${\sum_{v\in\V}\Mcomp{v}{c} = \qcomp{c}}$ for all candidates  $c\in\C$
}
We say that candidate $a$ \emph{admits} the~\pq-matching $\Mmat$ if $\Mcomp{v}{c}>0$ only if $(v,c)\in E_a$.

\begin{lemma}[\citealp{PluralityVeto}, Theorem 2] \label{lem:fracveto}
	Given an election \election along with ${\pvec\in\Delta_{\V}}$ and ${\qvec\in\Delta_{\C}}$, 
	\fracveto of \citet{PluralityVeto} returns a candidate admitting a \pq-matching.
\end{lemma}

\fracveto initializes the residual vote of each voter $v$ to $\pcomp{v}$ and the residual score of each candidate $c$~to~$\qcomp{c}$.
While some voter has positive residual vote, it chooses any~such voter, finds that voter's least-preferred candidate among those with positive residual score, and decreases both residuals by their minimum.
It returns the last candidate whose residual score reaches zero.

A $(\pvec^{\mathrm{uni}},\qvec^{\mathrm{plu}})$-matching is called a \emph{plurality matching},
where $\pcomp{v}^{\uni}=1/n$ for all ${v\in\V}$, and
$\qcomp{c}^{\plu}=\plu(c)/n$ for all ${c\in\C}$.
Equivalently, candidate $a$ admits a plurality matching if there is a bijection $M : \V \to \V$ such that $a \succeq_v \top(M(v))$ for every $v\in\V$.
With~these~weights, \fracveto \linebreak specializes to \pluveto, 
and achieves distortion 3, the lowest possible distortion that any social choice function can have \cite{anshelevich:bhardwaj:elkind:postl:skowron}.


\section{Known Weights}\label{sec:known-weights}

In this section, we consider the first information regime, where the weight profile $\wmat$ is given as input to the \linebreak social welfare function. 
At position~$j$, the entries of $\wcol{j}$ can be therefore viewed as voter masses in a weighted single-winner election over the remaining candidates.
This suggests constructing the ranking by recursively applying  a weighted extension of a social choice function.

The difficulty lies in composing these single-winner calls. If $\sigma^*=(c_1^*,\ldots,c_m^*)$ is a globally optimal ranking, then $c_j^*$ may already have been selected when position~$j$ is reached. 
Thus, a guarantee relative to the best available candidate is only local and need not be preserved by recursion.

We first study this problem in a black-box manner. 
If $f$ has distortion at most $\beta$, we consider its natural clone-weighted recursive extension  and show that it achieves distortion at most $3\beta$. 
We also show the limitation of using only the per-round guarantee: 
choosing a $\beta$-approximate candidate at every position can produce
a ranking whose cost is $2\beta$ times optimal.

This naturally raises the question whether producing \linebreak an~entire ranking requires distortion higher than $3$ which is the optimal single-winner bound.\footnote{This bound remains a lower bound in our setting, since setting $\wrow{v}=(1,0,\ldots,0)$ for every voter $v$ reduces the objective to the single-winner problem.}
It does not: exploiting \linebreak the matching certificates underlying \fracveto, 
we~show that Recursive \fracveto achieves the \linebreak optimal distortion of $3$.
Moreover, by operating directly on normalized masses, the rule runs in
polynomial time, whereas the black-box extension may require exponential
time because its explicit clone population can be exponential in the
bit length of the weights.

Finally, we study the analogous question for \copeland.
While its single-winner distortion is~$5$, \emph{Weighted Recursive} \linebreak \copeland requires a substantially different analysis, for which we obtain a distortion bound of~7. 
However, whether a~bound of $5$ is actually achievable, remains an open question.


\subsection{General Weighted Recursion}


We first formalize the \emph{clone-weighted}   extension of a social choice function.
For a rational voter-mass vector $\pvec\in\Delta_{\V}$, 
let $f_\pvec(\elec)$ be the outcome of $f$ after replacing 
every voter~$v$ by $L\pcomp{v}$ identical copies, 
where $L$ is the smallest positive \linebreak integer for which all multiplicities are integral.

Given a known weight profile $\wmat$, 
\emph{Weighted Recursive $f$} constructs a ranking $\sigma_f=(c_1,\ldots,c_m)$ from left to right. 
At~position~$j$, let $W_j:=\sum_{v\in\V}\wcomp{v}{j}$. If $W_j>0$, define
\[
  \pcomp{v}^{(j)}:=\frac{\wcomp{v}{j}}{W_j} \qquad \text{for all } v \in \V.
\]
The rule then applies $f_{\pvec^{(j)}}$ to the election restricted to $A_j$,  
places the resulting winner $c_j$ in position~$j$, and continues with the remaining candidates.
If $W_j=0$, monotonicity \linebreak implies that all remaining positions have zero weight, 
so the~remaining candidates are appended arbitrarily. 
We use this same convention for every recursive rule in this section.

For the remainder of this section, 
fix an arbitrary election \election,
a metric $d\cons\prof$, 
and a monotone weight \linebreak profile $\wmat\in\wdom$. 
The following rule-independent lemma will be used throughout this section.

\begin{lemma}\label{lem:top-cost}
For every ranking $\sigma^*=(c_1^*,\ldots,c_m^*) \in \sdom{\C}$,
\[
  \sum_{j=1}^m
  \sum_{v\in\V}\wcomp{v}{j}d(v,\top_j(v))
  \leq \cost(\sigma^*)
\]
\end{lemma}

\begin{proof}
	Fix a voter $v$, and let $d_v^{(1)}\leq\cdots\leq d_v^{(m)}$ denote the~distances from $v$ to the candidates in increasing order.
	Before position $j$, only $j-1$ candidates have been removed, so at least one of the $j$ closest candidates to $v$ is available.
	Therefore, $d(v,\top_j(v))\leq d_v^{(j)}$.
	Moreover, since $c_1^*,\ldots,c_k^*$ are distinct candidates, for every
	$k\in\{1,\ldots,m\}$,
	\[
	  \sum_{i=1}^{k}d(v,\top_i(v))
	  \leq
	  \sum_{i=1}^{k}d_v^{(i)}
	  \leq
	  \sum_{i=1}^{k}d(v,c_i^*).
	\]
	
	Set $\wcomp{v}{m+1}:=0$ and define $\drop{v}{\ell} :=\wcomp{v}{\ell}-\wcomp{v}{\ell+1}\ge0$. \linebreak
	Using the above bound and summation by parts, we obtain
	\begin{align*}
	  \sum_{j=1}^m \wcomp{v}{j}d(v,\top_j(v))
	  &= \sum_{\ell=1}^m \drop{v}{\ell} \sum_{j=1}^{\ell}d(v,\top_j(v))\\
	  &\le \sum_{\ell=1}^m \drop{v}{\ell} \sum_{j=1}^{\ell}d(v,c_j^*)\\
	  &= \sum_{j=1}^m \wcomp{v}{j}d(v,c_j^*).
	\end{align*}
	
	Summing this inequality over all voters and exchanging the order of summation completes the proof.
\end{proof}

\begin{theorem}\label{thm:generic-recursion}
  If social choice function $f$ has distortion at most $\beta$, then Weighted Recursive
  $f$ has distortion at most~$3\beta$.
\end{theorem}

\begin{proof}
	Let $\sigma_f =(c_1,\ldots,c_m)$ be the ranking returned by Weighted Recursive $f$, 
	and let $\sigma^*=(c_1^*,\ldots,c_m^*)$ be \linebreak an optimal ranking under $d$ and
	$\wmat$.
	
	When $W_j>0$, the cloned election has candidate costs proportional
	to $\poscost{j}{\cdot}$.
	Therefore, the distortion guarantee of $f$ gives
	\begin{equation} \label{generic:eq1}
		\poscost{j}{c_j} \leq \beta\min_{c\in A_j} \poscost{j}{c}.
	\end{equation}

	Since $\top_j(u)\in A_j$ for every voter $u$,
	\begin{equation} \label{generic:eq2}
		\min_{c\in A_j}\poscost{j}{c}
		\leq\frac1{W_j}\sum_{u\in\V} \wcomp{u}{j} \poscost{j}{\top_j(u)}.
	\end{equation}
	
	For $u,v\in\V$,  the triangle inequality  gives
	\[
		d(u,\top_j(v))
		\leq d(u,c_j^*)+d(c_j^*,v)+d(v,\top_j(v)).
	\]
	Thus, fixing $v$, multiplying by $\wcomp{u}{j}$, and summing over all voters $u\in\V$, we obtain
	\[
		\poscost{j}{\top_j(v)} 
		\leq \poscost{j}{c_j^*} + W_j d(c_j^*,v) + W_j d(v,\top_j(v)).
	\]
	Averaging this inequality over $v$ yields
	\begin{align} 
		\frac1{W_j}\sum_{v\in\V} \wcomp{v}{j} \poscost{j}{\top_j(v)}  \notag
		&\leq 2 \poscost{j}{c_j^*}\\
		&\quad+\sum_{v\in\V}\wcomp{v}{j}d(v,\top_j(v)). \label{generic:eq3}
	\end{align}
		
	Consequently,
\begin{align*}
		\poscost{j}{c_j}
		&\leq \beta\min_{c\in A_j}\poscost{j}{c} \tag{by \cref{generic:eq1}} \\
		&\leq \frac{\beta}{W_j}
			\sum_{v\in\V}\wcomp{v}{j}
			\poscost{j}{\top_j(v)} \tag{by \cref{generic:eq2}}\\
		&\leq 2\beta\,\poscost{j}{c_j^*}\\
		&\quad+\beta \sum_{v\in\V}\wcomp{v}{j}d(v,\top_j(v)). \tag{by \cref{generic:eq3}}
	\end{align*}
	When $W_j=0$, the weight column $\wcol{j}$ is identically zero, so both sides vanish and the same inequality holds trivially.
	
	Summing over all positions, and then applying Lemma~\ref{lem:top-cost}, we obtain
	\begin{align*}
		\cost(\sigma_f)
		&\leq2\beta \cost(\sigma^*)
		+\beta\sum_{j=1}^m\sum_{v\in\V}
		  \wcomp{v}{j}d(v,\top_j(v))\\
		&\le 3\beta\cost(\sigma^*),
	\end{align*}
	which completes the proof.
\end{proof}

Call a recursive construction
\emph{locally $\beta$-approximate} if
\[
\poscost{j}{c_j}
  \leq\beta\min_{c\in A_j}\poscost{j}{c}
\]
at every position with $W_j>0$.

\begin{proposition}\label{prop:2beta-lower}
  For every $\beta \ge 1$, locally $\beta$-approximate \linebreak recursion can
  incur a cost ratio of $2\beta$.
\end{proposition}

\begin{proof}
  Consider two voters $u,v$ and three candidates $a,b,c$ on the real line, 
  placed at $u=a=0$, $v=b=1$, and $c=\beta+1$.
  Thus, $u$ ranks ${a \succ b \succ c}$, 
  while $v$ ranks ${b \succ a \succ c}$.
  Give the voters monotone weight vectors 
  \[\wrow{u}=(\beta,0,0) \quad \text{and} \quad \wrow{v}=(1,1,0).\]
  
  At the first position, $\poscost{1}{a}=1$ and $\poscost{1}{b}=\beta$, 
  so selecting $b$ is locally $\beta$-approximate.
  Once $b$ is removed, $\poscost{2}{a}=1$ and
  $\poscost{2}{c}=\beta$, so selecting $c$ is locally $\beta$-approximate.
  The resulting ranking $(b,c,a)$ has cost $2\beta$, whereas $(a,b,c)$ is optimal with cost~$1$.
\end{proof}

Proposition~\ref{prop:2beta-lower} does not imply that Weighted Recursive~$f$ has
distortion at least $2\beta$ for every function $f$ of distortion~$\beta$: 
in the construction above, $f$ may select $a$ rather than $b$ and still have distortion~$\beta$.
Rather, it shows that the per-round cost comparison alone cannot yield a guarantee below $2\beta$. 
This means that any upper bound better than $2 \beta$ requires additional structural properties of the underlying rule. 
\pagebreak

The clone construction is likewise conceptual rather than computationally efficient.
After clearing denominators, a~voter of integer weight $K$ is represented by $K$ clones. 
Since $K$ requires only $O(\log K)$ bits to encode, explicitly constructing the cloned election can take time exponential in the input size. 
For \pluveto, we avoid this blowup \linebreak using the direct fractional extension \fracveto, which \fkreplace{runs in polynomial time and achieves thethe optimal distortion of~$3$}{results in a social welfare function with polynomial running time and optimal distortion~$3$.} 


\subsection{Recursive \fracveto}


For every position~$j$ with $W_j>0$, we define
\begin{align*}
  \pcomp{v}^{(j)} &:=\frac{\wcomp{v}{j}}{W_j} \quad \text{for all } v \in \V \\
  \qcomp{c}^{(j)} &:=\sum_{\substack{v\in\V:\\\top_j(v)=c}}\pcomp{v}^{(j)} \quad  \text{for all } c \in A_j.
\end{align*}

Instead of cloning the voters, we run \fracveto with 
normalized weights $\pvec^{(j)} \in \Delta_V$ and $\qvec^{(j)} \in \Delta_{A_j}$ 
on the~election restricted to $A_j$, place the winning candidate~$c_j$ in position~$j$,
and continue with remaining candidates. \linebreak
This implements Weighted Recursive \pluveto
without cloning the voters and thus runs in polynomial time.

We now analyze the rule.
Fix an election \election, a metric $d\cons\prof$, and a monotone
weight profile $\wmat\in\wdom$, and let
$\sigma=(c_1,\ldots,c_m)$ be the returned ranking.

The first lemma uses the \pq-matching guarantee of \fracveto to compare the selected candidate with any \emph{metric point}, 
by which we mean a point that can be adjoined to the metric space $(\V\cup A_j,d)$ without violating the~triangle inequality.
In particular, this includes candidates removed in earlier iterations.

\begin{lemma}\label{lem:metric_point}
	For every position $j$ and every metric point $x$,
	\begin{align*}
	  \poscost{j}{c_j}
	  &\leq 2\,\poscost{j}{x}
	  + \sum_{v\in\V}\wcomp{v}{j}d(v,\top_j(v)).
	\end{align*}
\end{lemma}

\begin{proof}
	The claim is immediate when $W_j=0$, so suppose that $W_j>0$.
	By Lemma~\ref{lem:fracveto}, the candidate placed at position~$j$ admits a $(\pvec^{(j)},\qvec^{(j)})$-matching.
	Scaling this matching by $W_j$ gives a matrix $\Mmat$ satisfying
	\begin{align*}
	  \sum_{c\in A_j}\Mcomp{v}{c}
	  &=\wcomp{v}{j}
	  &&\text{for every }v\in\V, \\
	  \sum_{v\in\V}\Mcomp{v}{c}
	  &=\sum_{\substack{v \in \V: \\ \top_j(v)=c}}\wcomp{v}{j} 
	  &&\text{for every }c\in A_j.
	\end{align*}

	As $c_j$ admits $\Mmat$, the condition $\Mcomp{v}{c}>0$ implies ${c_j\succeq_v c}$, and therefore $d(v,c_j)\leq d(v,c)$.
	Hence, using the~row and column sums of $\Mmat$, and the triangle inequality twice, \linebreak
	we~obtain \pagebreak
	\begin{align*}
		\poscost{j}{c_j}
		&= \sum_{v \in \V} \wcomp{v}{j} \, d(v, c_j) \\
		&= \sum_{v \in \V} \sum_{c \in A_j} \Mcomp{v}{c} \, d(v, c_j) \\ 
		&\leq \sum_{v\in\V}\sum_{c \in A_j} \Mcomp{v}{c} \, d(v,c) \\ 
		&\le \poscost{j}{x}
			+ \sum_{v\in\V}\sum_{c \in A_j} \Mcomp{v}{c} \, d(x, c) \\ 
		&= \poscost{j}{x}
		 +\sum_{v\in\V}\wcomp{v}{j}d(x,\top_j(v)) \\ 
		&\leq 2\,\poscost{j}{x}
		+\sum_{v\in\V}\wcomp{v}{j}d(v,\top_j(v)). 
	\end{align*}
	Since $x$ is an arbitrary metric point, this completes the proof.
\end{proof}

\begin{theorem}\label{thm:main}
	For every monotone weight profile, Recursive \fracveto has distortion at most~$3$.
	This bound is \linebreak optimal among weight-aware social welfare functions.
\end{theorem}

\begin{proof}
	Suppose that $\sigma^*=(c_1^*,\ldots,c_m^*) \in \sdom{\C}$ is an optimal ranking under $d$ and $\wmat$.
	For every position $j$, candidate $c_j^*$ is a metric point, even if it is no longer available when $c_j$ is selected.
	We may therefore apply Lemma~\ref{lem:metric_point} with $x=c_j^*$. 

	Summing the resulting inequalities over all positions and applying Lemma~\ref{lem:top-cost}, we obtain 
	\begin{align*}
		\cost(\sigma)
		&=\sum_{j=1}^m \poscost{j}{c_j}\\
		&\leq 2\,\cost(\sigma^*)
		+\sum_{j=1}^m\sum_{v\in\V}
		  \wcomp{v}{j}d(v,\top_j(v))\\
		&\le 3\,\cost(\sigma^*).
	\end{align*}
	
	Since the election, metric, and weight profile fixed above were arbitrary, 
	this completes the proof.
	Tightness follows by setting $\wrow{v}=(1,0,\ldots,0)$ for each voter ${v \in \V}$, 
	which reduces the objective to the single-winner~case.
\end{proof}


\subsection{Weighted Recursive \copeland}


\copeland\ selects a candidate of maximum out-degree in the
\emph{majority tournament}, the directed graph whose vertices are the
candidates with an edge from candidate~$a$ to candidate~$b$ whenever a strict majority of voters prefer $a$ to $b$.

For every position~$j$ with $W_j>0$, Weighted Recursive \linebreak \copeland 
constructs the \emph{weighted majority tournament}
$T_j=(A_j,E_j)$  as follows.
For each pair of distinct candidates $a,b\in A_j$, orient the edge
from $a$ to $b$ whenever
\[
  \sum_{v:a\succ_v b}\wcomp{v}{j}>\frac{W_j}{2}.
\]
The masses on the two sides sum to $W_j$.
Thus, exactly one orientation satisfies the strict inequality unless
both masses equal $W_j/2$, in which case we orient the edge arbitrarily. \linebreak
Consequently, $T_j$ is a tournament.
Weighted Recursive \copeland places a candidate $c_j$ of maximum
out-degree in $T_j$ in position~$j$ and continues.\pagebreak

Fix an election \election, a metric $d\cons\prof$, and a monotone
weight profile $\wmat\in\wdom$, and let
$\sigma=(c_1,\ldots,c_m)$ be the returned ranking.
We first bound the increase in social cost along a single edge of a
weighted majority tournament.
As in Lemma~\ref{lem:metric_point}, the reference point need not correspond to
an~available candidate, or even to a candidate at all.

\begin{lemma}\label{lem:majority_metric_point}
	For every position $j$ with $W_j>0$, every edge $(a,b)\in E_j$, and
	every metric point $x$,
	\[
	\poscost{j}{a}
	\leq
	\poscost{j}{b}
	+2\,\poscost{j}{x}.
	\]
\end{lemma}

\begin{proof}
	Let $\delta(v):=d(v,a)-d(v,b)$, and partition the voters into
	$V^+:=\setbuild{v\in\V}{\delta(v)>0}$ and $V^-:=\V\setminus V^+$.
	Every voter in $V^+$ ranks $b$ above $a$, while every voter who ranks
	$a$ above $b$ belongs to $V^-$.
	Set $s:=\sum_{v\in V^+}\wcomp{v}{j}$
	and $t:=\sum_{v\in V^-}\wcomp{v}{j}.$
	Since $(a, b) \in E_j$, the total mass of voters
	ranking $a$ above $b$ is at least the total mass of those ranking $b$
	above $a$, and thus, ${s\leq t}$.
	
	If $s=0$, then
	$\sum_{v\in\V}\wcomp{v}{j}\delta(v)\leq0$, and the claim follows.
	Hence, suppose that $s>0$, which also implies that $t>0$.
	
	For every $v\in V^+$ and $u\in V^-$, we have
	\begin{align*}
		\delta(v)
		&\leq \delta(v)-\delta(u)\\
		&=
		\bigl(d(v,a)-d(u,a)\bigr)
		+\bigl(d(u,b)-d(v,b)\bigr)\\
		&\leq 2d(v,u)\\
		&\leq 2d(v,x)+2d(u,x).
	\end{align*}
	Averaging this inequality over $u\in V^-$ with weights
	$\wcomp{u}{j}/t$, multiplying by $\wcomp{v}{j}$, and summing over
	$v\in V^+$ gives
	\begin{align*}
		\poscost{j}{a}
		-\poscost{j}{b}
		&=\sum_{v\in\V}\wcomp{v}{j}\delta(v)\\
		&\leq\sum_{v\in V^+}\wcomp{v}{j}\delta(v)\\
		&\leq
		2\sum_{v\in V^+}\wcomp{v}{j}d(v,x)\\
		&\quad+
		2\frac{s}{t}\sum_{u\in V^-}\wcomp{u}{j}d(u,x)\\
		&\leq 2\,\poscost{j}{x},
	\end{align*}
	where the last inequality follows from $s\leq t$.
\end{proof}

A candidate belongs to the \emph{uncovered set} if, for every other
candidate $b$, it either defeats $b$ directly or defeats some candidate
that defeats $b$. It is well known that every \copeland winner belongs
to the uncovered set (see, for example,
\cite{anshelevich:bhardwaj:elkind:postl:skowron}).
Combining this property with
Lemma~\ref{lem:majority_metric_point} gives the following corollary.

\begin{corollary}\label{cor:copeland_metric_point}
  For every position $j$ with $W_j>0$, every candidate $b\in A_j$, and
  every metric point $x$,
  \[
    \poscost{j}{c_j}
    \leq
   \poscost{j}{b}
    +4\,\poscost{j}{x}.
  \]
\end{corollary}

Taking $x=b$ in Corollary~\ref{cor:copeland_metric_point} recovers the usual
\mbox{factor-$5$} comparison for a single \copeland winner.
At position $j$ of the recursive rule, however, the candidate occupying
that position in an optimal ranking may no longer be available. \linebreak
We~therefore use it as the metric point $x$ 
and average the available comparison candidate $b$ over voters' top choices.

\begin{theorem}\label{thm:cope7}
  For every monotone weight profile, Weighted Recursive \copeland has distortion
  at most~$7$.
\end{theorem}

\begin{proof}
	Let $\sigma^*=(c_1^*,\ldots,c_m^*)\in\sdom{\C}$ be an optimal ranking
	under $d$ and $\wmat$.
	Fix a position $j$ with ${W_j>0}$ and a voter $u\in\V$.
	Since $\top_j(u)\in A_j$, we may apply
	Corollary~\ref{cor:copeland_metric_point} with $b=\top_j(u)$ and $x=c_j^*$.
	\fkedit{Notice that $c_j^*$ is a valid metric point even if it was removed in an earlier iteration.}
	Thus,
	\[
		\poscost{j}{c_j}
		\leq 4\,
		  \poscost{j}{c_j^*}
		+\poscost{j}{\top_j(u)}.
	\]
	Averaging this inequality over voters $u$ with probabilities
	$\wcomp{u}{j}/W_j$, and writing
	\[
	B_j:=
	\frac{1}{W_j}\sum_{u\in\V}
	\wcomp{u}{j}
	\poscost{j}{\top_j(u)},
	\]
	gives
	\begin{equation}\label{eq:copeland_average}
	\poscost{j}{c_j}
	\leq
	4\,
	\poscost{j}{c_j^*}
	+B_j. \tag{$*$}
	\end{equation}
	
	For every pair of voters $u,v\in\V$, the triangle inequality gives
	\[
	d(v,\top_j(u))
	\leq
	d(v,c_j^*)+d(u,c_j^*)+d(u,\top_j(u)).
	\]
	Consequently,
	\begin{align*}
		B_j &= \frac{1}{W_j}\sum_{u,v\in\V} \wcomp{u}{j}\wcomp{v}{j}d(v,\top_j(u))\\
		&\leq 2\,\poscost{j}{c_j^*}
		+\sum_{u\in\V}\wcomp{u}{j}d(u,\top_j(u)).
	\end{align*}
	Here, the first two terms produced by the triangle inequality each equal
	$\poscost{j}{c_j^*}$,
	while the third gives the final
	sum.
	Substituting this bound into \cref{eq:copeland_average}, we obtain
	\begin{align*}
		\poscost{j}{c_j}
		&\leq 6\,\poscost{j}{c_j^*}
		+\sum_{v\in\V}\wcomp{v}{j}d(v,\top_j(v)).
	\end{align*}
	The same inequality is immediate when $W_j=0$.
	
	Since Lemma~\ref{lem:top-cost} is independent of the rule used to select the
	candidates, summing over all positions yields
	\begin{align*}
		\cost(\sigma)
		&\leq 6\,\cost(\sigma^*)
		+
		\sum_{j=1}^m\sum_{v\in\V}
		  \wcomp{v}{j}d(v,\top_j(v))\\
		&\leq 7\,\cost(\sigma^*).
	\end{align*}
	
	Since the election, metric, and weight profile fixed above were arbitrary, 
	this completes the proof.
\end{proof}

When $\wrow{v}=(1,0,\ldots,0)$ for every voter $v \in \V$, 
only the first position contributes to the objective and Weighted Recursive \copeland reduces to ordinary \copeland.
The~single-winner lower bound of~$5$
\cite{anshelevich:bhardwaj:elkind:postl:skowron} therefore continues to
apply.
Thus, we currently only know that the distortion of Weighted Recursive \copeland
lies between~$5$ and~$7$.

\section{Identical Weights}\label{sec:identical-weights}

We now consider the second information regime, in which 
the weight profile is unknown, but all voters share the same monotone weight vector. 
Despite not observing this vector, we show that the optimal distortion of~$3$ can still be achieved.

While the~previous section required properties specific to
\fracveto and ~\copeland, here the distortion guarantee of every social
choice function will transfer to its recursive extension.

This regime subsumes both single-winner voting and committee selection.
The shared vector $\wvec=(1,0,\ldots,0)$ recovers the single-winner
objective, whereas
\[
  \wvec^{(k)}
  :=(\underbrace{1,\ldots,1}_{k\text{ entries}},0,\ldots,0)
\]
recovers the sum-of-costs objective for a committee of size~$k$.
\citet[Theorem~8]{GoelHK18} show that recursively applying 
a social choice function \linebreak preserves its distortion for these cutoff vectors.
Our result generalizes theirs to every monotone shared weight vector
and gives a sharper guarantee as its range decreases.


Write $\wrow{v}=\wvec=(\wcomp{1},\ldots,\wcomp{m})$ for every $v\in\V$,
where $\wcomp{1}\geq\cdots\geq\wcomp{m}\geq0$ and $\wcomp{1}>0$.
Let $\range{\wvec}:={1-\wcomp{m}/\wcomp{1}}\in[0,1]$ denote the
\emph{normalized range} of $\wvec$, which is invariant under scaling.

\begin{theorem}\label{thm:identical-ub}
	If $f$ has distortion at most $\beta\geq1$, 
	then under every shared weight vector $\wvec$, Recursive $f$ has distortion 
	at~most $1+(\beta-1)\cdot\range{\wvec}$.
\end{theorem}

\begin{proof}
	Fix an election \election, a metric $d\cons\prof$, 
	and a shared weight vector $\wvec \in \widom$.
	Let
	$\sigma=(c_1,\ldots,c_m)$ be the ranking returned by Recursive $f$.
	Index the candidates so~that
	\[
		\cost(c_1^*)
		\leq\cdots\leq
		\cost(c_m^*).
	\]
	Since $\wcomp{1}\geq\cdots\geq\wcomp{m}$,
	the ranking $\sigma^*:=(c_1^*,\ldots,c_m^*)$ is optimal under $\wvec$.
	When position~$j$ is filled, at least one of
	$c_1^*,\ldots,c_j^*$ remains available.
	Hence, the guarantee of $f$ on the restricted election gives
	\[
		\cost(c_j)
		\leq\beta\min_{c\in A_j}\cost(c)
		\leq\beta\,\cost(c_j^*).
	\]

	Define 
	\[
	\begin{aligned}
	  X&:=\sum_{c\in\C}\cost(c),\\
	  Y&:=\sum_{j=1}^m
	    \bigl(\wcomp{j}-\wcomp{m}\bigr)\cost(c_j^*).
	\end{aligned}
	\]
	For every position~$j$, write
	$\wcomp{j}=\wcomp{m}+(\wcomp{j}-\wcomp{m})$.
	The first term assigns the same weight $\wcomp{m}$ to every position.
	Since every ranking contains each candidate exactly once,
	\[
		\sum_{j=1}^m\wcomp{m}\cost(\sigma_j)
		=\wcomp{m}\sum_{c\in\C}\cost(c)
		=\wcomp{m}X.
	\]
	Therefore,
	\begin{align*}
		\cost(\sigma)
		&\leq\wcomp{m}X+\beta Y,\\
		\cost(\sigma^*)
		&=\wcomp{m}X+Y.
	\end{align*}
	Moreover, monotonicity gives
	$0\leq\wcomp{j}-\wcomp{m}\leq\wcomp{1}-\wcomp{m}$ for every position~$j$.
	Since $\sigma^*$ is a complete ranking, $c_1^*,\ldots,c_m^*$ contain
	every candidate exactly once.
	Therefore, we obtain
	\begin{align*}
		0\leq Y
		&=\sum_{j=1}^m
			\bigl(\wcomp{j}-\wcomp{m}\bigr)\cost(c_j^*)\\
		&\leq
			\bigl(\wcomp{1}-\wcomp{m}\bigr)
			\sum_{j=1}^m\cost(c_j^*)\\
		&=\bigl(\wcomp{1}-\wcomp{m}\bigr)
			\sum_{c\in\C}\cost(c)\\
		&=\bigl(\wcomp{1}-\wcomp{m}\bigr)X
		=\wcomp{1}\range{\wvec}X.
	\end{align*}
	
	If the optimal cost is zero, then
	$\wcomp{m}X+Y=0$, and the preceding bound
	$\cost(\sigma)\leq\wcomp{m}X+\beta Y$
	shows that the cost of $\sigma$ is also zero.
	Otherwise, $\wcomp{m}X+Y>0$.
	Since the function
	$x\mapsto x/(\wcomp{m}X+x)$ is nondecreasing for $x\geq0$, the bound on $Y$ gives
	\begin{align*}
		\frac{Y}{\wcomp{m}X+Y}
		&\leq
		\frac{(\wcomp{1}-\wcomp{m})X}
		     {\wcomp{m}X+(\wcomp{1}-\wcomp{m})X}\\
		&=
		\frac{\wcomp{1}-\wcomp{m}}{\wcomp{1}}
		=\range{\wvec}.
	\end{align*}
	Consequently,
	\begin{align*}
		\distid(\text{Recursive } f) &= \frac{\cost(\sigma)}{\cost(\sigma^*)} \\
		&\leq
		\frac{\wcomp{m}X+\beta Y}
		     {\wcomp{m}X+Y}\\
		&=1+(\beta-1)
			\frac{Y}{\wcomp{m}X+Y}\\
		&\leq1+(\beta-1)\range{\wvec},
	\end{align*}
	which completes the proof.
\end{proof}


The single-winner vector witnesses the hardest case: 
its~normalized range is~$1$, and Recursive $f$ has the same distortion as~$f$.
As the vector becomes flatter, the guarantee improves, reaching~$1$ for a uniform vector.

We next give a range-dependent lower bound that applies to every social welfare function.

\begin{theorem}\label{thm:identical-lb}
	For every social welfare function $F$ and \linebreak every $r \in [0, 1]$, 
	there exists an election $\elec$ and an identical weight profile $\wvec \in \widom$ with $\range{\wvec} = r$, under which the distortion is at least 
	$\dist_{\elec, \wvec}(F(\elec)) \ge \frac{4-r}{4-3r} = 1+\frac{2r}{4-3r}$.
\end{theorem} 

\begin{proof}
	Consider two voters $u,v$ and two candidates $a,b$, with
	$a\succ_u b$ and $b\succ_v a$, and take the shared vector $\wvec=(1,1-r)$.
	Suppose that $F$ returns $(a,b)$; the other case is symmetric.
	Place $a$ at~$0$, voter $u$ at~$1$, and voter $v$ and candidate $b$ at~$2$ on the real line.
	This metric is consistent and gives $\cost(a\cond d)=3$ and $\cost(b\cond d)=1$.
	Hence, the returned ranking costs $4-r$, whereas the optimal ranking $(b,a)$ costs $4-3r$.
	Their ratio is $(4-r)/(4-3r)$, as claimed. 
\end{proof}
\pagebreak

The lower bound applies to every social welfare function, 
whereas \cref{thm:identical-ub} with $\beta=3$ gives the upper bound $1+2r$ for Recursive \pluveto.
The two bounds coincide at~$1$ when $r=0$ and at~$3$ when $r=1$;
and for intermediate ranges, their additive gap is
\[
  \left(1+2r\right)-\frac{4-r}{4-3r}
  =
  \frac{6r(1-r)}{4-3r}
  \leq\frac{2}{3},
\]
with equality at $r=2/3$.
Thus, the upper bound is exact at both extremes and is no more than
$2/3$ above the universal lower bound in \cref{thm:identical-lb}.

\begin{corollary}\label{cor:identical-pv}
	Recursive \pluveto has optimal distortion~$3$ under identical weight profiles.
	More precisely, for every election $\elec$ 
	and identical weight profile ${\wvec \in \widom}$,
	$F=$~Recursive \pluveto has distortion at most
	$\dist_{\elec,\wmat} \bigl(F(\elec)\bigr) \leq 1+2\range{\wvec}$.
\end{corollary}

\section{Normalized Weights}\label{sec:normalized-weights}

We now consider the third information regime, 
in which the weight profile remains hidden and may vary across voters, but every voter's
weight vector is normalized.
We consider the unit-sum and unit-top domains defined in the preliminaries.
In contrast to identical weights, hidden heterogeneity makes constant
distortion impossible under either normalization.

\begin{theorem}\label{thm:normalized-lb}
  $\distsum(F)\geq\frac{m}{2}$ and $\disttop(F)\geq\frac{m}{4}$ 
  for every social welfare function $F$.
\end{theorem}

\begin{proof}
	Let $h:=m/2$, and partition the candidates into sets $X$ and $Y$ of
	size~$h$.
	Consider two voters $u$ and $v$, where $u$ ranks every candidate in
	$X$ above every candidate in $Y$, while $v$ has the reverse
	preference.
	Place $u$ and all candidates in $X$ at one point, and $v$ and all
	candidates in $Y$ at another point at distance~$1$.
	This metric is consistent with the rankings.
	
	Let $\sigma:=F(\elec)$.
	We first consider unit-sum weights. \linebreak
	Suppose that $\sigma_1\in Y$.
	Give $u$ the top-only vector and give $v$ weight $1/h$ on each of
	the first $h$ positions and 0 thereafter.
	The cost of $\sigma$ is at least~$1$, whereas a ranking that places
	one candidate from $X$ first and $h-1$ candidates from $Y$ next has
	cost $1/h$.
	Thus, the distortion is at least $h=m/2$.
	If $\sigma_1\in X$, the same argument holds after exchanging $u$ and $v$.
	
	For unit-top weights, consider the first $h$ positions of $\sigma$.
	Suppose that at least $h/2$ of these positions contain candidates from $Y$.
	Give $u$ weight~$1$ on the first $h$ positions and zero thereafter,
	and give $v$ the top-only vector.
	The returned ranking costs at least $h/2$, whereas a ranking that
	places all candidates in $X$ first costs~$1$.
	The ratio is therefore at least $h/2=m/4$.
	If fewer than $h/2$ candidates are from $Y$, the~symmetric argument
	applies with the roles of $X$ and $Y$ exchanged.
\end{proof}


We now turn to upper bounds, and show that
Recursive \linebreak \pluveto achieves distortion $O(m)$ under both normalizations.

\begin{theorem} \label{thm:normalized-ub}
	Recursive \pluveto has distortion at most $8m+3$ under unit-sum weight profiles
	and at most $4m+1$ under unit-top weight profiles.
\end{theorem}

\begin{proof}
	We first consider the unit-sum domain.
	Fix an election \linebreak \election, 
    a metric $d\cons\prof$, and a profile
	$\wmat\in\wsumdom$.
	Let $\sigma$ be the ranking returned by Recursive \pluveto, \linebreak
    let $\sigma^*=(c_1^*,\ldots,c_m^*)$ be an optimal ranking. 
	Define $\bar w_j:=\frac1n\sum_{v\in\V}\wcomp{v}{j}$ 
    and ${\barcost(\pi):=\sum_{j=1}^m\bar w_j \cost(\pi_j)}$,
	and let $\Pi:=\frac1n\sum_{u,v\in\V}d(u,v)$.

    For every ranking $\pi\in \sdom{\C}$, the triangle inequality gives, for every
    pair of voters $u,v\in\V$ and every position $j$,
    \[
    d(v,\pi_j)\le d(v,u)+d(u,\pi_j).
    \]
    Since this inequality holds for every $u$, averaging over all
    $u\in\V$ yields
    \[
    d(v,\pi_j)
    \le
    \frac1n\sum_{u\in\V} d(v,u)
    +
    \frac1n\sum_{u\in\V} d(u,\pi_j).
    \]
    Multiplying both sides by $\wcomp{v}{j}$ and summing over all voters~$v$ 
    and positions $j$, we obtain
    \begin{align*}
    \cost(\pi)
    &=
    \sum_{j=1}^m\sum_{v\in\V}
    \wcomp{v}{j}d(v,\pi_j)\\
    &\le
    \frac1n
    \sum_{j=1}^m
    \sum_{u,v\in\V}
    \wcomp{v}{j}d(v,u)
    +
    \frac1n
    \sum_{j=1}^m
    \sum_{u,v\in\V}
    \wcomp{v}{j}d(u,\pi_j).
    \end{align*}
    For the first term, exchanging the order of summation and using the
    unit-sum normalization,
    \[
    \sum_{j=1}^m \wcomp{v}{j}=1
    \qquad\text{for every }v\in\V,
    \]
    gives
    \[
    \frac1n
    \sum_{u,v\in\V}
    d(u,v)
    \sum_{j=1}^m\wcomp{v}{j}
    =
    \frac1n
    \sum_{u,v\in\V}d(u,v)
    =\Pi.
    \]
    For the second term, exchanging the order of summation gives 
    \begin{align*}
    \frac1n
    \sum_{j=1}^m
    \sum_{u,v\in\V}
    \wcomp{v}{j}d(u,\pi_j)
    &=
    \sum_{j=1}^m
    \left(
    \frac1n\sum_{v\in\V}\wcomp{v}{j}
    \right)
    \sum_{u\in\V}d(u,\pi_j)\\
    &=
    \sum_{j=1}^m
    \bar w_j
    \sum_{u\in\V}d(u,\pi_j)
    =
    \barcost(\pi).
    \end{align*}
    Therefore,
    \[
    \cost(\pi)\le \Pi+\barcost(\pi).
    \]
    
	Averaging the reverse inequality
	$d(u,\pi_j)\leq d(u,v)+d(v,\pi_j)$ similarly gives
	$\barcost(\pi)\leq\cost(\pi)+\Pi$.
    
	Due to monotonicity and unit-sum normalization, we have
	${\wcomp{v}{1}\geq 1/m}$.
	Consequently,
	$\cost(\sigma^*)
	\geq\frac1m\cost(c_1^*)$.
	The triangle inequality therefore gives
	\begin{align*}
		\Pi
		&\leq\frac1n\sum_{u,v\in\V}
			\bigl(d(u,c_1^*)+d(v,c_1^*)\bigr)\\
		&=2\cost(c_1^*)\\
		&\leq2m\cost(\sigma^*).
	\end{align*}

	Since the average vector $(\bar w_1,\ldots,\bar w_m)$ is also monotone,
	\cref{thm:identical-ub} gives
	$\barcost(\sigma)\leq3\min_{\pi\in\sdom{\C}}\barcost(\pi)$.
	Thus,
	\[
		\cost(\sigma)
		\leq3\barcost(\sigma^*)+\Pi
		\leq3\cost(\sigma^*)+4\Pi
		\leq(8m+3)\cost(\sigma^*).
	\]

	We next consider the unit-top domain.
	Fix an election \election, a metric $d\cons\prof$, and a profile
	$\wmat\in\wtopdom$.
	Let $\sigma=(c_1,\ldots,c_m)$ be the returned ranking and let
	$\sigma^*=(c_1^*,\ldots,c_m^*)$ be an optimal ranking.
	
	By definition, there exists a bijection $M_j:\V\to\V$, for every position $j$,
	such that $c_j\succeq_v\top_j(M_j(v))$ for all $v\in\V$.
	Fix $v\in\V$ and set $u:=M_j^{-1}(v)$.
	Since $d\cons\prof$, we obtain
	\begin{align*}
		d(v,c_j)
		&\leq d(v,u)+d(u,c_j)\\
		&\leq d(v,u)+d(u,\top_j(v))\\
		&\leq 2d(v,u)+d(v,\top_j(v)).
	\end{align*}
	
	Let $D:=\sum_{j=1}^m\sum_{v\in\V}
	\wcomp{v}{j}d\bigl(v,M_j^{-1}(v)\bigr)$.
	Multiplying the above inequality by $\wcomp{v}{j}$ and summing
	over all voters and positions, we obtain
	\begin{align*}
		\cost(\sigma)
		&\leq
		  2D+\sum_{j=1}^m\sum_{v\in\V}
		  \wcomp{v}{j}d(v,\top_j(v))\\
		&\leq
		  2D+\cost(\sigma^*),
	\end{align*}
	where the second inequality follows from Lemma~\ref{lem:top-cost}.
	
	It remains to bound $D$.
	Unit-top normalization implies $\wcomp{v}{j}\leq1$ and
	$\sum_j\wcomp{v}{j}\leq m$.
	Therefore,
	\begin{align*}
		D
		&\leq \sum_{v\in\V}
			\left(\sum_{j=1}^m\wcomp{v}{j}\right)d(v,c_1^*)\\
		&\quad+\sum_{j=1}^m\sum_{v\in\V}
			\wcomp{v}{j}d\bigl(c_1^*,M_j^{-1}(v)\bigr)\\
		&\leq m\cost(c_1^*)
			+\sum_{j=1}^m\sum_{v\in\V}
			d\bigl(c_1^*,M_j^{-1}(v)\bigr)\\
		&=2m\cost(c_1^*),
	\end{align*}
	where the last equality follows because $M_j^{-1}$ is a permutation
	of $\V$ for every~$j$.
	Finally, since $\wcomp{v}{1}=1$ for every voter $v$, we have
	\[
		\cost(\sigma^*)
		\geq\cost(c_1^*).
	\]
	Hence,
	$D\leq2m\cost(\sigma^*)$,
	and consequently
	\[
		\cost(\sigma)
		\leq(4m+1)
	\cost(\sigma^*),
	\]
	completing the proof.
\end{proof}

Together with \cref{thm:normalized-lb}, these bounds determine the optimal distortion up to constant factors: it is $\Theta(m)$ under either normalization.

\pagebreak

\section{Conclusion}

We initiated the study of metric distortion for social welfare functions, 
extending the metric distortion framework from selecting a single winner to constructing complete rankings under position-weighted objectives. 
Our results show that the amount of information available about the positional weights fundamentally determines the achievable distortion: \linebreak
knowing the weights allows optimal constant distortion, identical hidden weights still admit strong guarantees, while heterogeneous hidden weights inevitably lead to linear distortion, even under natural normalizations.

These results leave several interesting directions for \linebreak future work. 
In particular, can the distortion of Recursive \copeland be improved from the current upper bound of~$7$ to the conjectured bound of $5$? 
More generally, it would be interesting to understand which single-winner distortion guarantees extend to recursive ranking rules, 
and to study metric distortion for social welfare functions under richer preference models, such as limited cardinal information or randomized social welfare functions.




\paragraph{Acknowledgments}

This work was supported in part by the Ministry of National Education of the Republic of T\"urkiye. 
OpenAI's GPT-5.5 and GPT-5.6 Sol were used to assist in developing some of the results presented in this paper. 
The authors reviewed and verified the resulting arguments and proofs and take full responsibility for the paper's content.

\bibliography{references,bibliography,publications}

\end{document}